\documentclass{llncs}
\usepackage{graphicx}
\usepackage{amssymb}
\usepackage{amsmath}
\usepackage{color}
\usepackage{times}
\usepackage{verbatim}
\usepackage{url}
\usepackage{tikz}
\usetikzlibrary{chains, shapes.misc}
\tikzset{
    nonterminal/.style={
      rectangle,
      minimum size=6mm,
      very thick,
      draw=red!50!black!50,
      top color=white,
      bottom color=red!50!black!20,
      font=\itshape},
    terminal/.style={
      rounded rectangle,
      minimum size=6mm,
      very thick,draw=black!50,
      top color=white,bottom color=black!20,
      font=\ttfamily}
    }
\usepackage{wrapfig}
\usetikzlibrary{arrows,decorations.pathmorphing,backgrounds,positioning,fit,petri}

\usepackage{subcaption}
\usepackage{subfig}

\def\C{\mathcal{C}}	
\newcommand{\bigO}{\ensuremath{\mathcal{O}}} 

\def\I{\mathcal{I}}

\def\uri{\mathbf{U}}
\def\bn{\mathbf{B}}
\def\lit{\mathbf{L}}
\def\const{\mathbf{C}}
\def\var{\mathbf{V}}

\def\local{\textrm{local}}
\def\lclosure{\textsf{lclosure}}

\title{ Query Answering over Contextualized RDF Knowledge with Forall-Existential Bridge Rules: 
Attaining Decidability using Acyclicity (full version)
}
\author{
Mathew Joseph$^{1,2}$  \and Gabriel Kuper$^2$ \and Luciano Serafini$^1$
}
\institute{ 
$^1$ DKM, FBK-IRST, Trento, Italy \\
$^2$ DISI, University Of Trento, Trento, Italy\\
\medskip
\email{\{mathew,serafini\}@fbk.eu, kuper@disi.unitn.it}
}

\begin{document}
\maketitle
\begin{abstract}
The recent outburst of context-dependent knowledge on the Semantic Web (SW) has led to the 
realization  of the importance of the quads in the SW community. Quads, which extend a
standard RDF triple, by adding a new parameter of the `context' of an RDF triple, 
thus informs a reasoner to distinguish between the knowledge in various 
contexts. Although this distinction separates
 the triples in an RDF graph into various contexts, and allows the reasoning to be decoupled across various contexts,
  bridge rules need to be provided for inter-operating the knowledge across these contexts.
  We call a set of quads together with the bridge rules, a quad-system. 
In this paper, we discuss the problem of query answering 
over quad-systems with expressive forall-existential bridge rules. 
It turns out the query answering over quad-systems is undecidable, in general. 
We derive a decidable class of quad-systems,
namely \emph{context-acyclic} quad-systems, for which query answering can be done 
using forward chaining. Tight bounds for data and combined complexity of query entailment has been established 
for the derived class.  
\end{abstract}


\vspace{-5pt}
\section{Introduction}
One of the major recent changes in the SW community is the transformation from a \emph{triple} 
 to a \emph{quad} as its primary knowledge carrier. As a consequence, more and more triple stores are becoming \emph{quad} stores.
Some of the popular quad-stores are 4store\footnote{http://4store.org}, Openlink Virtuoso\footnote{http://virtuoso.openlinksw.com/rdf-quad-store/}, and some
of the current popular triple stores like Sesame\footnote{http://www.openrdf.org/}
internally keep track of the context by storing arrays of four names $(c,s,p,o)$ (further denoted as $c:(s,p,o)$),
 where $c$ is an identifier that stands for the context of the triple $(s,p,o)$.  Some of the recent initiatives
in this direction have also extended existing formats like N-Triples to N-Quads. 
The latest Billion triples challenge datasets (BTC 2012) have all been released in the N-Quads format.

One of the main benefits of quads over triples are that they allow users to specify various attributes of meta-knowledge 
that further qualify knowledge~\cite{caroll-bizer-named-graphs-data-provenance-wwww2005}, and also allow users to query for this
meta knowledge~\cite{QueryingMetaknowledge}. Examples of these attributes, which are also called 
\emph{context dimensions}~\cite{cyc_context_space},
are provenance, creator, intended user, creation time, validity time, geo-location, and topic. Having defined various contexts
in which triples are dispersed, one can declare in another meta-context $mc$, statements such as 
$mc\colon(c_1$, \text{creator}, \text{John}$)$, $mc\colon(c_1$, \text{expiryTime}, \text{``jun-2013''}$)$ that talk 
about the knowledge in context $c_1$, in this case its creator and expiry time.
Another benefit of such a contextualized approach is that it opens possibilities of interesting ways for
querying a contextualized knowledge base. For instance, if context $c_1$ contains knowledge about Football World Cup 2014
and context $c_2$ about Football Euro Cup 2012. Then the query ``who beat Italy in both Euro Cup 2012 and
World Cup 2014'' can be formalized as the conjunctive query: 
\[
 c_1\text{: }(x,\textrm{beat},\textrm{Italy}) \wedge c_2\text{: }(x,\textrm{beat},\textrm{Italy}), \text{where $x$ is a variable}.
\] 
As the knowledge can 
be separated context wise and simultaneously be fed to separate reasoning engines, 
this approach increases both efficiency and scalability. Besides the above flexibility,  
  \emph{bridge rules}~\cite{DDL} can be provided for inter-interoperating the knowledge in different 
  contexts. Such rules are primarily of the form:
\begin{eqnarray}\label{eqn:bridgeRule}
c:\phi(\vec x) \rightarrow c':\phi'(\vec x) \nonumber
\end{eqnarray}
where $\phi,\phi'$ are both atomic concept (role) symbols, $c,c'$ are contexts. The semantics of such a rule is that if, for any $\vec a$, 
$\phi(\vec a)$ holds in context $c$, then $\phi'(\vec a)$ should hold in context $c'$, 
where $\vec a$ is a unary/binary vector dependending on whether $\phi, \phi'$ are concept/role symbols.
Although such bridge rules serve the purpose of specifying  knowledge interoperability from a source 
context $c$ to a target context $c'$, in many practical situations there is the need of 
interoperating  multiple source contexts with multiple target targets, for which the bridge rules of the form (\ref{eqn:bridgeRule})
is inadequate. Besides, one would also want the ability of creating new values in target contexts for the bridge rules.


In this work, we consider \emph{forall-existential bridge rules} 
that allows conjunctions and existential quantifiers in them, 
and hence is more expressive than those, in DDL~\cite{DDL} and 
McCarthy et al.~\cite{McCarthy95formalizingcontext}. 
A set of quads together with such bridge rules is called a \emph{quad-system}.  
\noindent The main contributions of this work can be summarized as:
\begin{enumerate}
\item We provide a basic semantics for contextual reasoning over quad-systems, and study contextualized 
conjunctive query answering over them. For query answering, we use the notion of a \emph{distributed chase},
which is an extension of a standard \emph{chase}~\cite{JohnsonK84,AbiteboulHV95} that is widely used in databases and 
KR for the same.
 \item We show that conjunctive query answering
 over quad-systems, in general,  is undecidable. 
 We derive a class of quad-systems called \emph{context acyclic} quad-systems, for which query answering is decidable and can be 
 done by forward chaining. We give both data and combined complexity of conjunctive query entailment for the same. 
\end{enumerate}
The paper is structured as follows: In section 2, we formalize the idea of contextualized quad-systems, giving
various definitions and notations for setting the background.
In section 3, we formalize the query answering on quad-systems, define notions such as distributed chase
that is further used for query answering, and give the undecidability results of query entailment for
unrestricted quad-systems. In section \ref{sec:cAcyclic}, we present
context acyclic quad-systems and its properties. We give an account of relevant related works in
section \ref{sec:related work}, and 
conclude in section \ref{sec:conclusion}.


\section{Contextualized Quad-Systems}
 \vspace{-5pt} In this section, we formalize the notion of a
 quad-system and its semantics.  For any vector or sequence $\vec x$, we denote
 by $\|\vec x\|$ the number of symbols in $\vec x$, and by $\{\vec x\}$ the set of symbols
 in $\vec x$. For any sets $A$ and $B$, $A \rightarrow B$ denotes the
 set of all functions from set $A$ to set $B$.
Given  the set of URIs $\uri$, the set of blank nodes $\bn$, and  the set of literals $\lit$, the set 
$\const=\uri \uplus \bn \uplus \lit$ are called the set of (RDF) constants. Any $(s,p,o) \in \const \times 
\const \times \const$ is called a generalized RDF triple (from now on,
just triple). A graph is defined as a set of triples.
A \emph{Quad} is a tuple of the form $c\colon(s,p,o)$, where $(s,p,o)$ is a triple and $c$ is a URI\footnote{Although,
in general a context identifier can be a constant, for the ease of notation, we restrict them to be a URI}, 
called the \emph{context identifier} that denotes the context of the RDF triple.
A \emph{quad-graph} is defined as a set of quads.
For any quad-graph $Q$ and any context identifier $c$, we denote by $graph_Q(c)$ the set $\{(s,p,o)|c\colon(s,p,o) \in Q\}$.
We denote by $Q_{\C}$ the quad-graph whose set of context identifiers is $\C$. 
 Let $\var$ be the set of variables, any element of the set ${\const}^{\var}=\var \cup \const$ is a \emph{term}. 
 Any $(s,p,o) \in {\const}^{\var} \times {\const}^{\var} \times
 {\const}^{\var}$ is called a \emph{triple pattern}, and an expression
 of the form $c\colon(s,p,o)$, where $(s,p,o)$ is a triple pattern,
 $c$ a context identifier, is called a \emph{quad pattern}.  A triple
 pattern $t$, whose variables are elements of the vector $\vec x$ or elements of the vector 
 $\vec y$ is written as $t(\vec x, \vec y)$.  For any function $f\colon A \rightarrow B$, the \emph{restriction} of $f$ to a set $A'$,
is the mapping $f|_{A'}$ from $A' \cap A$ to $B$ s.t. $f|_{A'}(a)=f(a)$, for each $a \in A \cap A'$.
 For any triple pattern $t=(s,p,o)$ and a function $\mu$ from  $\var$ to a set $A$,
 $t[\mu]$ denotes $(\mu'(s), \mu'(p), \mu'(o))$, where $\mu'$ is an extension of $\mu$ to $\const$ s.t. $\mu'|_{\const}$
is the identity function. For any set of  triple patterns $G$, $G[\mu]$ denotes $\bigcup_{t \in G} t[\mu]$.
For any vector of constants $\vec a=\langle a_1,\dots,a_{\|\vec a\|}\rangle$,
 and vector of variables $\vec x$ of the same length, $\vec x/\vec
 a$ is the function $\mu$ s.t. 
 $\mu(x_i)=a_i$, for $1\leq i\leq \|\vec a\|$. We use the notation $t(\vec a, \vec y)$ to
 denote $t(\vec x, \vec y)[\vec x/\vec a]$. 
\vspace{-10pt}
\paragraph{Bridge rules (BRs)} \emph{Bridge rules (BR)} enables knowledge propagation across contexts.
Formally, a BR is an expression of the form:
 \begin{eqnarray}
\label{eqn:intraContextualRuleWithQuantifiers}
  \forall \vec x \forall \vec z \ [c_1\text{: }t_1(\vec x, \vec z) \wedge ... \wedge c_n\text{: }t_n(\vec x, \vec z) 
 \rightarrow \exists \vec y \ c'_1\text{: }t'_1(\vec x, \vec y) \wedge ... \wedge c'_m\text{: }t'_m(\vec x, \vec y)]
\end{eqnarray}  
where $c_1, ..., c_n, c'_1,...,c'_m$ are context identifiers, 
$\vec x, \vec y, \vec z$ are vectors of variables s.t. $\{\vec x\}, \{\vec y\}$,
and $\{\vec z\}$ are pairwise disjoint. $t_1(\vec x$, $\vec z)$, 
..., $t_n(\vec x$, $\vec z)$ are triple patterns which do not contain blank-nodes, and whose set of
variables are from $\vec x$ or $\vec z$. $t'_1(\vec x$, $\vec y)$, ...,$t'_m(\vec x, \vec y)$ are triple patterns, whose set of
variables are from $\vec x$ or $\vec y$, and also does not contain blank-nodes. 
For any BR, $r$, of the form (\ref{eqn:intraContextualRuleWithQuantifiers}), 
$body(r)$ is the set of quad patterns $\{c_1\text{: }t_1(\vec x$, $\vec z)$,...,$c_n\text{: }t_n(\vec x$, $\vec z)\}$,
and $head(r)$ is the set of quad patterns $\{c'_1\text{: }t'_1(\vec x, \vec y)$, ...  $c'_m\text{: }t'_m(\vec x$, $\vec y)\}$.
\begin{definition}[Quad-System]
A \emph{quad-system} $QS_{\C}$ is defined as a pair $\langle Q_{\C}, R\rangle$, where $Q_{\C}$ is a quad-graph, whose
set of context identifiers is $\C$, and $R$ is a set
of BRs.
\end{definition}
\vspace{-5pt}
 For any quad-graph  $Q_{\C}$ (BR $r$), its symbols size $\|Q_{\C}\|$ $(\|r\|)$ is the number of symbols
 required to print $Q_{\C}$ $(r)$. Hence, $\|Q_{\C}\|\approx 4*|Q_{\C}|$, where $|Q_{\C}|$ denotes the cardinality of
 the set $Q_{\C}$. Note that $|Q_{\C}|$ equals the number of quads in $Q_{\C}$. 
 For a BR $r$, $\|r\|\approx 4*k$, where $k$ is the number of quad-patterns in $r$.
 For a set of BRs $R$, its size $\|R\|$ is given as $\Sigma_{r \in R} \|r\|$. 
 For any quad-system $QS_{\C}$ $=$ $\langle Q_{\C},R \rangle$, its size $\|QS_{\C}\|$ $=$ $\|Q_{\C}\|+\|R\|$. 
\paragraph{Semantics}
In order to provide a semantics for enabling reasoning over a quad-system, 
we need to use a local semantics for each context to interpret the knowledge pertaining to it. 
Since the primary goal of this paper is a decision procedure for query answering over quad-systems based on 
forward chaining, we consider the following desiderata for the choice of the local semantics:
\begin{itemize}
 \item there exists a set of inference rules and an operation $\lclosure()$ 
 that computes the deductive closure of a graph w.r.t to the local semantics using the inference rules.
 \item given a finite graph as input, the $\lclosure()$ operation, terminates with a finite graph as output in polynomial time
 whose size is polynomial w.r.t. to the input set. 
\end{itemize}
Some of the alternatives for the local semantics satisfying the above mentioned criterion
 are  Simple, RDF, RDFS~\cite{Hayes04rdfsemantics}, OWL-Horst~\cite{terHorst200579} etc. 
 Assuming that a local semantics has been fixed, for any context $c$, 
 we denote by $I^c=\langle\Delta^c, \cdot^c\rangle$ an interpretation structure for the local semantics, 
 where $\Delta^c$ is the interpretation domain, $\cdot^c$ the corresponding interpretation function.
 Also $\models_{\local}$ denotes the local satisfaction relation between a local interpretation structure and
 a graph.
Given a quad graph $Q_{\C}$, a \emph{distributed interpretation structure} is an 
indexed set $\I^{\C}=\{I^c\}_{c\in\C}$, 
 where $I^c$ is a local interpretation structure, for each $c \in \C$.
We define the satisfaction relation $\models$ between a distributed interpretation
structure $\I^{\C}$ and a quad-system $QS_{\C}$ as: 
\begin{definition}[Model of a Quad-System]\label{def:model-quad-system}
 A distributed interpretation structure $\I^{\C}=\{I^c\}_{c\in\C}$  satisfies
 a quad-system $QS_{\C}$ $=$ $\langle Q_{\C}$, $R\rangle$, in symbols $\I^{\C}$ $\models$ $QS_{\C}$, iff all the following conditions are satisfied:
 \vspace{-5pt}
\begin{enumerate}
\item\label{item:graphSatisfaction}  
$I^c \models_{\local} graph_{Q_{\C}}(c)$, 
for each $c \in \C$;
\item\label{item:rigid} 
$a^{c_i}=a^{c_j}$, for any $a \in \const$, 
$c_i, c_j \in \C$; 
 \item\label{item:BRSatisfaction} 
 for each BR $r \in R$ of the form~(\ref{eqn:intraContextualRuleWithQuantifiers}) and
for each $\sigma \in \var \rightarrow \Delta^{\C}$, where $\Delta^{\C}$ $=$ $\bigcup_{c\in \C} \Delta^c$, if 
\vspace{-5pt}
\[
I^{c_1} \models_{\local} t_1(\vec x, \vec z)[\sigma], ..., I^{c_n} \models_{\local} 
t_n(\vec x, \vec z)[\sigma], 
\]
then there exists function $\sigma'\supseteq \sigma$, s.t.
\vspace{-5pt}
 \[
I^{c'_1} \models_{\local} t'_1(\vec x, \vec y)[\sigma'], ..., I^{c'_m} \models_{\local} t'_m(\vec x, \vec y)[\sigma'].
 \]
\end{enumerate}
\end{definition}
\vspace{-5pt}
Condition \ref{item:graphSatisfaction} in the above definition ensures that for any model $\I^{\C}$
of a quad-graph, each $I^c \in \I^{\C}$ is a local model of
the set of triples in context $c$. Condition \ref{item:rigid} ensures that
any constant $c$ is rigid, i.e. represents the same resource across a quad-graph, irrespective of the context in which it occurs.
Condition \ref{item:BRSatisfaction} ensure that any model of a quad-system satisfies each BR in it.
Any $\I^{\C}$ s.t. $\I^{\C} \models QS_{\C}$ is said to be a model of $QS_{\C}$.
A quad-system $QS_{\C}$ is said to be \emph{consistent} if there exists a model $\I^{\C}$, s.t. 
$\I^{\C} \models QS_{\C}$, and otherwise said to be \emph{inconsistent}. 
For any quad-system $QS_{\C}=\langle Q_{\C},R \rangle$, it can be the case that $graph_{Q_{\C}}(c)$ is locally consistent,
 for each $c \in \C$, whereas $QS_{\C}$ is not consistent. This is because the set of BRs $R$ adds more 
knowledge to the quad-system, and restricts the set of models that satisfy the quad-system. 
\begin{definition}[Quad-system entailment]
(a) A quad-system $QS_{\C}$ entails a quad $c\colon(s$, $p$, $o)$, in symbols $QS_{\C}$ $\models$ $c\colon(s,p,o)$,
iff for any distributed interpretation structure
$\I^{\C}$, if $\I^{\C} \models QS_{\C}$ then $\I^{\C} \models \langle \{c\colon(s,p,o)\}, \emptyset \rangle$.
 (b) A quad-system $QS_{\C}$ entails a quad-graph $Q'_{\C'}$, in symbols $QS_{\C} \models Q'_{\C'}$ iff $QS_{\C} \models c\colon(s,p,o)$
for any $c\colon(s,p,o) \in Q'_{\C'}$. 
(c)  A quad-system $QS_{\C}$ entails a BR $r$ iff 
for any 
$\I^{\C}$, if $\I^{\C} \models QS_{\C}$ then $\I^{\C} \models \langle \emptyset, \{r\} \rangle$. 
(d) For a set of BRs
$R$, $QS_{\C} \models R$ iff $QS_{\C} \models r$, for every $r \in R$. 
(e)  Finally, a quad-system $QS_{\C}$ entails another 
quad-system $QS'_{\C'}$ $=$ $\langle Q'_{\C'}, R' \rangle$, in symbols $QS_{\C} \models QS'_{\C'}$ iff $QS_{\C} \models Q'_{\C'}$ and
$QS_{\C} \models R'$. 
\end{definition}
\vspace{-5pt}
\noindent We call  the decision problems (DPs) corresponding to the entailment problems (EPs) in  (a), (b), (c), (d), and (e) as
\emph{quad EP, quad-graph EP, BR EP, BRs EP, and quad-system EP}, respectively.
\vspace{-5pt}
\vspace{-5pt}
\section{Query Answering on Quad-Systems}\label{sec:Query}
In the realm of quad-systems, the classical conjunctive queries or select-project-join queries are 
slightly extended to what we call \emph{Contextualized
Conjunctive Queries} (CCQs). A CCQ $CQ(\vec x)$ is an expression of the form:
\begin{equation}\label{eqn:CCQ}
 \exists \vec y \ q_1(\vec x, \vec y) \wedge ... \wedge q_p(\vec x, \vec y)
\end{equation}
where $q_i$, for $i=1,...,p$ are quad patterns over vectors of \emph{free variables} $\vec x$ 
and \emph{quantified variables} $\vec y$. A CCQ is called a boolean CCQ if it does not have any free variables.
For any CCQ $CQ(\vec x)$ and a vector $\vec a$ of constants s.t. $\|\vec x\|=\|\vec a\|$, $CQ(\vec a)$ is boolean.  
A vector $\vec a$ is an \emph{answer} for a CCQ $CQ(\vec x)$ 
w.r.t. structure $\I_{\C}$, 
in symbols $\I_{\C} \models CQ(\vec a)$, iff there exists assignment $\mu\colon \{\vec y\}\rightarrow  \bn$ s.t.
$\I_{\C} \models \bigcup_{i=1, \ldots, p} q_i(\vec a, \vec y)[\mu]$. 
A vector $\vec a$ is a \emph{certain answer} for a CCQ $CQ(\vec x)$ over a quad-system $QS_{\C}$, iff
$\I_{\C} \models CQ(\vec a)$, for every model $\I_{\C}$ of $QS_{\C}$.
 Given a quad-system $QS_{\C}$, a CCQ $CQ(\vec x)$, and a vector $\vec a$,
 DP of determining whether $QS_{\C} \models CQ(\vec a)$ is called the \emph{CCQ EP}.
It can be noted that the other DPs over quad-systems that we have seen are reducible to CCQ EP. Hence,
in this paper, we primarily focus on the CCQ EP. 

\paragraph{dChase of a Quad-System}
In order to do query answering over a quad-system, we employ what has been called
in the literature, a \emph{chase}~\cite{JohnsonK84,AbiteboulHV95}, specifically,  
we adopt notion of 
the \emph{skolem chase} given in Marnette~\cite{Marnette:2009} and Cuenca Grau et al~\cite{Bernardo:dlacyclicity:2012}.
In order to fit the framework of quad-systems, we extend the standard notion of chase to a 
 \emph{distributed chase}, abbreviated \emph{dChase}. In the following, we show how the dChase of a quad-system can be constructed.
 
For any BR $r$ of the form (\ref{eqn:intraContextualRuleWithQuantifiers}), 
the \emph{skolemization} $sk(r)$ is the result of replacing each $y_i \in \{\vec y\}$ with a
globally unique Skolem function $f_i^r$, s.t. 
$f_i^r\colon$ $\const^{\|\vec x\|}$ $\rightarrow$ $\bn_{sk}$, where $\bn_{sk}$ is a fresh set of blank nodes
called \emph{skolem blank nodes}. Intuitively, for every distinct 
vector $\vec a$ of constants, with $\|\vec a\|=\|\vec x\|$, 
$f_i^r(\vec a)$ is a fresh blank node, whose node id is a hash of $\vec a$. 
Let $\vec f^r=\langle f_1^r, ..., f_{\|\vec y\|}^r \rangle$ be a vector of distinct Skolem functions;
for any BR $r$ the form (\ref{eqn:intraContextualRuleWithQuantifiers}), with slight abuse (Datalog notation) we write its 
skolemization $sk(r)$ as follows:
\vspace{-5pt}
\begin{eqnarray}\label{eqn:intraContextualRule}
 c_1:t_1(\vec x, \vec z), ..., c_n\colon t_n(\vec x, \vec z)   
 \rightarrow c'_1\colon t'_1(\vec x, \vec f^r), ..., c'_m\colon t'_m(\vec x, \vec f^r)
\end{eqnarray}
Moreover, a skolemized BR $r$ of the form  (\ref{eqn:intraContextualRule}) can be replaced
by the following semantically equivalent set of formulas, whose symbol size is worst case quadratic w.r.t $\|r\|$:
\begin{eqnarray}
&& \{c_1\colon t_1(\vec x, \vec z), ..., c_n\colon t_n(\vec x, \vec z) \rightarrow c'_1\colon t'_1(\vec x, \vec f^r), \label{eqn:BRsingleHead} \\ 
&& ..., \nonumber \\
&&  c_1\colon t_1(\vec x, \vec z), ..., c_n\colon t_n(\vec x, \vec z) \rightarrow c'_m\colon t'_m(\vec x, \vec f^r) \}\nonumber
\end{eqnarray} 
Note that each BR in the above set has exactly one quad pattern with optional function symbols
in its head part. Also note that a BR with out function symbols can be replaced with a set of BRs with 
single quad-pattern heads. Hence, w.l.o.g, we assume that any BR in a skolemized set $sk(R)$ of BRs 
is of the form (\ref{eqn:BRsingleHead}).
For any quad-graph $Q_{\C}$ and a skolemized BR $r$ of the form~(\ref{eqn:BRsingleHead}), 
\emph{application} of $r$ on $Q_{\C}$, denoted by $r(Q_{\C})$, is given as:
\vspace{-5pt}
\[
r(Q_{\C})=
\bigcup_{\mu \in \var \rightarrow \const} \left \{ 
\begin{array}{l}
 c'_1:t'_1(\vec x, \vec f^r)[\mu] \ | \ c_1:t_1(\vec x,\vec z)[\mu]  \in Q_{\C}, ..., c_n:t_n(\vec x, \vec z)[\mu] \in Q_{\C}
\end{array}
\right\}
\]
For any set of skolemized BRs $R$, application of $R$ on $Q_{\C}$ is given by:
$ R(Q_{\C})$ $= $ $\bigcup_{r\in R} r(Q_{\C})$.
For any quad-graph $Q_{\C}$, we define:
\begin{eqnarray}
 \lclosure(Q_{\C})= \bigcup_{c \in \C}  
 \{c:(s,p,o) \ | (s,p,o) \in \lclosure(graph_{Q_{\C}}(c))\} \nonumber 
\end{eqnarray}
For any quad-system $QS_{\C}$ $=$ $\langle Q_{\C},R \rangle$, 
 \emph{generating BRs} $R_F$ is the set of BRs in $sk(R)$ with function symbols, and   
the \emph{non-generating BRs}  is the set $R_I=sk(R)\setminus R_F$. \\
Let $dChase_0(QS_{\C})$ $=$ $\lclosure(Q_{\C})$; for $i\in \mathbb{N}$, $dChase_{i+1}(QS_{\C})$ $=$
\[
 \begin{array}{l l}
 \lclosure(dChase_i(QS_{\C})\cup R_I(dChase_i(QS_{\C}))),  & \hspace{.5cm }\text{if } R_I(dChase_i(QS_{\C})) \not \subseteq  \\
 & \hspace{.5cm } dChase_i(QS_{\C});\\
 \lclosure(dChase_i(QS_{\C}) \cup R_F(dChase_i(QS_{\C}))), & \hspace{.5cm } \text{otherwise};
 \end{array}
\]
The dChase of $QS_{\C}$, denoted $dChase(QS_{\C})$, is given as:
\[
 dChase(QS_{\C})= \bigcup_{i\in \mathbb{N}} dChase_i(QS_{\C})
\]
Intuitively, $dChase_i(QS_{\C})$ can be thought of as the state of $dChase(QS_{\C})$ at the end of iteration $i$. 
It can be noted that, if there exists $i$ s.t. $dChase_i(QS_{\C})$ $=$ $dChase_{i+1}(QS_{\C})$, 
then $dChase(QS_{\C})$ $=$ $dChase_i(QS_{\C})$.
An iteration $i$, s.t. $dChase_i(QS_{\C})$ is computed by the application of the set of (resp. non-)generating 
BRs $R_F$ (resp. $R_I$), on $dChase_{i-1}(QS_{\C})$ is called a (resp. \emph{}non-)\emph{generating iteration}.
The dChase $dChase(QS_{\C})$ of a consistent quad-system $QS_{\C}$ 
 is a \emph{universal model}~\cite{Deutsch:2008} of the quad-system, i.e. it is a model of $QS_{\C}$, and 
 for any model $\I_{\C}$ of $QS_{\C}$, there is a homomorphism from $dChase(QS_{\C})$ to $\I_{\C}$. Hence, for any boolean CCQ $CQ()$, 
 $QS_{\C} \models CQ()$ iff there exists a map $\mu\colon \var(CQ) \rightarrow \const$ s.t. $\{CQ()\}[\mu]\subseteq dChase(QS_{\C})$.
We call the sequence $dChase_0(QS_{\C})$, $dChase_1(QS_{\C}), ...,$ the \emph{dChase sequence} of $QS_{\C}$. 
The following lemma shows that in a dChase sequence of a quad-system, the result of a single generating 
iteration and a subsequent number of non-generating iterations causes only an exponential blow up in size.
\begin{lemma}\label{lemma:chaseSizeIncrease}
 For a quad-system $QS_{\C}=\langle Q_{\C}, R \rangle$, the following holds:   
 (i) if $i \in \mathbb{N}$ is a generating iteration, then
$\|dChase_i(QS_{\C})\|= \bigO(\|dChase_{i-1}(QS_{\C})\|^{\|R\|})$,
(ii) suppose $i \in \mathbb{N}$ is a generating iteration, and for any $j \geq 1$, 
$i+1$, $...$, $i+j$ are non-generating iterations, then
$\|dChase_{i+j}(QS_{\C})\|= \bigO(\|dChase_{i-1}(QS_{\C})\|^{\|R\|})$,
(iii) for any iteration $k$, $dChase_k(QS_{\C})$ can be computed in time 
$\bigO(\|dChase_{k-1}(QS_{\C})\|^{\|R\|})$.
\end{lemma}
\begin{proof}[sketch] 
(i) $R$ can be applied on $dChase_{i-1}(QS_{\C})$ by grounding $R$
to the set of constants in $dChase_{i-1}(QS_{\C})$, 
the number of such groundings is of the order $\bigO($ $\|dChase_{i-1}(QS_{\C})\|^{\|R\|})$,
$\|R(dChase_{i-1}(QS_{\C}))\|$ $=$ $\bigO(\|R\|*\|dChase_{i-1}(QS_{\C})\|^{\|R\|})$.
Since $\lclosure$ only increases the size polynomially,
$\|dChase_i(QS_{\C})\|$ $=$ $\bigO($ $\|dChase_{i-1}($ $QS_{\C})\|^{\|R\|})$.

(ii) From (i) we know that $\|R(dChase_{i-1}(QS_{\C}))\|$ $=$ $\bigO(\|dChase_{i-1}(QS_{\C})\|^{\|R\|})$. Since, no new constant
is introduced in any subsequent non-generating iterations, and since any quad contains only four constants,
the set of constants in any subsequent dChase iteration is $\bigO(4*\|dChase_{i-1}(QS_{\C})\|^{\|R\|})$.
Since only these many constants can appear in positions $c,s,p,o$ of any quad generated in the 
subsequent iterations, the size of $dChase_{i+j}(QS_{\C})$ can only increase polynomially, which means
that $\|dChase_{i+j}($ $QS_{\C})\|$ $=$ $\bigO(\|dChase_{i-1}(QS_{\C})\|^{\|R\|})$.

(iii) Since any dChase iteration $k$ involves the following two operations: (a) $\lclosure()$, and
(b) computing $R(dChase_{k-1}(QS_{\C}))$. 
(a) can be done in PTIME w.r.t to its input. (b) can be 
done in the following manner: ground $R$ to the set of constants in $dChase_{i-1}(QS_{\C})$; then for each grounding $g$, 
 if $body(g)\subseteq dChase_{i-1}(QS_{\C})$, then add $head(g)$ to $R(dChase_{k-1}(QS_{\C}))$.
 Since, the number of such groundings is of the order $\bigO(\|dChase_{k-1}(QS_{\C})\|^{\|R\|})$,
 and checking if each grounding is contained in $dChase_{k-1}(QS_{\C})$, can be done in time polynomial in 
 $\|dChase_{k-1}(QS_{\C})\|$, the time taken for (b) is $\bigO(\|dChase_{k-1}(QS_{\C})\|^{\|R\|})$.
 Hence, any iteration $k$ can be done in time $\bigO(\|dChase_{k-1}(QS_{\C})\|^{\|R\|})$.
$\qed$
\end{proof}
\noindent 
Although, we now know how to compute the dChase of a quad-system, which can be used for deciding CCQ EP, 
it turns out that for the class of quad-systems whose BRs are of the form 
(\ref{eqn:intraContextualRuleWithQuantifiers}), which we call \emph{unrestricted quad-systems}, 
the dChase can be infinite.
This raises the question if there are other approaches that can be used, for instance
similar problem arises in DLs with value creation, due to the presence of existential quantifiers, 
whereas the approaches like the one in Glim et al.~\cite{GliHoLuSa-07} provides an algorithm for CQ entailment 
based on query rewriting.  
\begin{theorem}\label{theorem:undecidable}
 The  CCQ EP over unrestricted quad-systems is undecidable.
\end{theorem}
\vspace{-2pt}
\begin{proof}[sketch]
We show that the well known undecidable problem of 
non-emptiness of intersection of context-free grammars (CFGs) is reducible to the CCQ entailment
problem. Given two CFGs, $G_1=\langle V_1, T, S_1, P_1 \rangle$ and $G_2=\langle V_2, T, S_2, P_2 \rangle$, 
where $V_1, V_2$ are the set of variables, 
$T$ s.t. $T \cap (V_1 \cup V_2)=\emptyset$ is the set of terminals. $S_1 \in V_1$ is the start symbol of $G_1$,
and $P_1$ are the set of PRs of the form $v \rightarrow \vec w$, where $v \in V$, $\vec w$ is a sequence of the form
$w_1...w_n$, where $w_i \in V_1 \cup T$. Similarly $s_2, P_2$ is defined.  Deciding whether the language
generated by the grammars $L(G_1)$ and $L(G_2)$ have non-empty intersection is known to be 
undecidable~\cite{Harrison-formal-language}.

Given two CFGs $G_1=\langle V_1, T, S_1, P_1 \rangle$ and $G_2=\langle V_2, T, S_2, P_2 \rangle$,
we encode grammars $G_1, G_2$ into a quad-system $QS_c=\langle Q_{c},R\rangle$, with only a single context identifier $c$.
Each PR $r= v \rightarrow \vec w \in P_1 \cup P_2$, with $\vec w=w_1w_2w_3..w_n$, is encoded as a BR of the form:
$c\colon(x_1,w_1,x_2), c\colon(x_2, w_2, x_3),...,c\colon(x_n,w_n,x_{n+1}) 
\rightarrow c\colon(x_1,v,x_{n+1}), \nonumber$
where $x_1,..,x_{n+1}$ are variables. For each terminal symbol $t_i \in T$, $R$ contains
a BR of the form:
$ c\colon(x, \texttt{rdf:type}, C) \rightarrow \exists y \ c\colon(x,t_i,y), 
 c\colon(y, \texttt{rdf:type},C) \nonumber$
and $Q_c$ is the singleton:
$\{ c\colon(a,\texttt{rdf:type},C)\}$.
It can be observed that:
\begin{eqnarray}
 QS_c \models \exists y \ c\colon(a , S_1, y) \ \wedge \ c\colon(a, S_2, y) \Leftrightarrow 
L(G_1)\cap L(G_2)\neq \emptyset \nonumber
\end{eqnarray}
We refer the reader to Appendix for the complete proof.
$\qed$ 
\end{proof} 
\section{Context Acyclic Quad-Systems: A decidable class}\label{sec:cAcyclic}
In the previous section, we saw that query answering on unrestricted quad-systems is undecidable, in general. 
We in the following define a class of quad-systems 
for which query entailment is decidable. The class has the property that algorithms based on forward chaining, for deciding query entailment, 
 can straightforwardly be implemented (by minor extensions) on existing quad stores. It should be noted that the technique we propose is 
 reminiscent of the \emph{Weak acyclicity}~\cite{Fagin05dataexchange,Deutsch03reformulationof} technique used in the realm of Datalog+-.
Consider a BR $r$ of the form:
$
 c_1\colon t_1(\vec x, \vec z), c_2\colon t_2(\vec x, \vec z) \rightarrow
\exists \vec y \ c_3\colon t_3(\vec x, \vec y), c_4\colon t_4(\vec x$, $\vec y).
$
Since such a rule triggers propagation of knowledge in a quad-system, specifically triples from the source contexts
$c_1,c_2$ to the target contexts $c_3,c_4$ in a quad-system. 

\begin{wrapfigure}{r}{0.65\textwidth}
\centering
 \begin{tikzpicture} 
 \draw[-,xscale=1.3] (0,0) --  (1,0.5)  --  (2,0) -- node[above] {} (1,-0.5) -- (0,0) ;
 \draw[-,xscale=1.3,xshift=1cm,yshift=0.5cm] (0,0) -- node[below] {} (1,0.5) -- (2,0) -- (1,-0.5) -- (0,0) ;
 \draw[-,xscale=1.3, xshift=2cm,yshift=0cm] (0,0) -- node[below] {} (1,0.5) -- (2,0) -- (1,-0.5) -- (0,0) ;
 \draw[-,xscale=1.3, xshift=1cm,yshift=-0.5cm] (0,0) -- node[below] {} (1,0.5) -- (2,0) -- (1,-0.5) -- (0,0) ;
 \draw[-,xscale=1.3, xshift=-0.5cm,yshift=0.75cm] (0,0) -- node[below] {} (1,0.5) -- (2,0) -- (1,-0.5) -- (0,0) ;
 \draw[-,xscale=1.3, xshift=2.5cm,yshift=0.75cm] (0,0) -- node[below] {} (1,0.5) -- (2,0) -- (1,-0.5) -- (0,0) ;
 \draw [->] (0,0.8) to[out=45,in=135] 
 node[above]{
$
c_1\text{: }t_1(\vec x,\vec z),c_2\text{: }t_2(\vec x,\vec z)\rightarrow \exists \vec y
 \ c_3\text{: }t_3(\vec x,\vec y),c_4\text{: }t_4(\vec x,\vec y)  
$
} 
 (5.2,0.8);
 
 \draw [-] (0.7,0) to[out=65,in=235] (2,1.82);
 \draw [->] (3,1.86) to[out=-15,in=115] (4,0);
 
 \draw (0.8,0.6) node {$c_1$};
  \draw (1.2,-0.1) node {$c_2$};
   \draw (4.5,0.6) node {$c_3$};
\draw (4.2,-0.1) node {$c_4$};
\end{tikzpicture}
\caption{}
\label{fig:propogationRule}
\end{wrapfigure}
\noindent As shown in Fig.~\ref{fig:propogationRule},
we can view a BR as a propagation rule across distinct compartments of knowledge, divided as contexts. 
For any BR of the form~(\ref{eqn:intraContextualRuleWithQuantifiers}), each context in the 
set $\{c'_1,...,c'_m\}$ is said to depend on the set of contexts $\{c_1,...,c_n\}$.  
In a  quad-system 
$QS_{\C}=\langle Q_{\C}, R\rangle$, 
for any $r \in R$, of the form~(\ref{eqn:intraContextualRuleWithQuantifiers}), any context whose identifier is in the set
$\{c \ | \ c\colon(s,p,o) \in head(r), s\text{ or }p\text{ or }o$ is an existentially quantified variable$\}$, 
is called  a \emph{triple} \emph{generating} \emph{context} (TGC).
One can analyze the set of  BRs in a quad-system $QS_{\C}$ using a context
dependency graph, which is a directed graph, whose nodes are context identifiers in $\C$,
s.t. the nodes corresponding to TGCs are marked with a $*$, and whose edges are constructed as follows:
for each BR of the form~(\ref{eqn:intraContextualRuleWithQuantifiers}), 
there exists an edge from each $c_i$ to $c'_j$, for $i=1,...,n, j=1,...,m$. 
A quad-system is said to be 
\emph{context acyclic}, iff its context dependency graph does not contain cycles involving TGCs. 
\begin{example}\label{eg:dependency}
Consider a quad-system, whose  set of BRs $R$ are:
 \begin{eqnarray}
 c_1\colon (x_1, x_2, \mathsf{U_1}) \rightarrow \exists y_1 \ c_2\colon (x_1, x_2,y_1), 
 c_3\colon (x_2, \texttt{rdf:type}, \texttt{rdf:Property}) \label{eqn:BRinExample} \\
c_2\colon (x_1,x_2,z_1) \rightarrow c_1\colon (x_1, x_2, \mathsf{U_1})\hspace{6.5cm}\label{eqn:BRinExample2}  \\
c_3\colon (x_1, x_2, x_3) \rightarrow c_1\colon(x_1, x_2, x_3) \hspace{6.5cm}  \nonumber
 \end{eqnarray}
where $\mathsf{U_1}$ be a URI,
whose corresponding dependency graph is shown in Fig.~\ref{fig:dependency-graph}. Note that the node corresponding
to the triple generating context $c_2$ is marked with a `$*$' symbol. 
Since  the cycle $(c_1,c_2,c_1)$ in the quad-system 
contains $c_2$ which is a TGC, the quad-system is not context acyclic. 
\end{example}

\noindent In a context acyclic quad-system $QS_{\C}$, since there exists no cyclic path through any TGC node in the 
context dependency graph, there exists a set of TGCs $\C' \subseteq \C$ s.t. for any $c \in \C'$, there exists no 
incoming path\footnote{assume that paths have at least one edge} from a TGC to $c$. We call such TGCs, \emph{level-1 TGCs}. 
In other words, a TGC $c$ is a level-1 TGC, if for any $c'\in {\C}$,
there exists an incoming path from $c'$ to $c$, implies  $c'$ is not a TGC. 
For $l\geq 1$, a level-$l\text{+}1$ TGC $c$ is a TGC 
that has an incoming path from a level-$l$ TGC, and for any incoming path from a level-$l'$ TGC to $c$, is s.t. $l'\leq l$. 
Extending the notion of level also to the non-TGCs, we say that any non-TGC that does not have any incoming paths from a TGC
  \begin{wrapfigure}{l}{0.4\textwidth}
\centering
\begin{tikzpicture}
[place/.style={circle,draw=black!50,fill=white!20,thick,
inner sep=0pt,minimum size=6mm}
]
\node(c1) at ( 0,1.5) [place] {$c_1$};
\node(c2)[place,,label=below:$*$] at (-1.4,0) [place] {$c_2$};
\node(c3) at ( 1.4,0) [place] {$c_3$};
\draw [->] (c1.west) to[out=202.5,in=67.5] node[auto]{} (c2.north);
\draw [->] (c3.north) to[out=105,in=-22.5] node[auto]{} (c1.east);
\draw [->] (c2.east) to[out=22.5,in=247.5] node[auto,swap]{} (c1.south);
\draw [->] (c1.south) to[out=297.5,in=152.5] node[auto,swap]{} (c3.west);
\end{tikzpicture}
\vspace{-10pt}
\caption{Context\\ Dependency graph}
\label{fig:dependency-graph}
\end{wrapfigure}
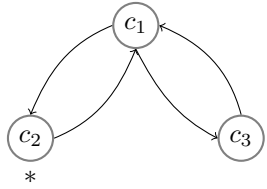
is at level-0; we say that any non-TGC $c \in \C$ is at level-$l$, if there exists an 
incoming path from a level-$l$ TGC to $c$, and for any incoming path from a level-$l'$ TGC to $c$, is s.t. $l'\leq l$.
Hence, the set of contexts in a context acyclic quad-system can be partitioned using the above notion of levels.
\begin{definition}
 For a quad-system $QS_{\C}$, a context $c \in \C$ is said to be
\emph{saturated} in an iteration $i$, iff for any quad of the form $c\colon(s,p,o)$,
$c\colon(s,p,o) \in dChase(QS_{\C})$ implies $c\colon(s,p,o) \in dChase_i(QS_{\C})$.
\end{definition}
Intuitively, context $c$ is saturated in the dChase iteration $i$, if no new quad of the form $c\colon(s,p,o)$ will be 
generated in any $dChase_k(QS_{\C})$, for any $k > i$. The following lemma  
gives the relation between the saturation of a context and the required number of dChase iterations, for a context acyclic
quad-system. 
\begin{lemma}\label{lemma:saturatedContext}
 For any context acyclic quad-system, the following holds: 
 (i) any level-0 context  is saturated before the first generating iteration,
 (ii) any level-1 TGC  is saturated after the first generating iteration,
 (iii) any level-$k$ context is saturated before the $k+1$th generating iteration.
\end{lemma}
\begin{proof}
Let $QS_{\C}=\langle Q_{\C},R \rangle$ be the quad-system, whose first generating iteration is $i$.
 
 (i) for any level-0 context $c$,  
 any BR $r \in R$, and any quad-pattern of the form $c\colon(s,p,o)$, if $c\colon(s,p,o) \in 
 head(r)$, then for any $c'$ s.t. $c'\colon(s',p',o')$ occurs in $body(r)$ implies that $c'$ is a level-0 context and
 $r$ is a non-generating BR. Also, since $c'$ is a level-0 context, the same applies to $c'$. Hence, it turns out 
 that only non-generating BRs can bring triples to any level-0 context. Since at the end of iteration $i-1$, $dChase_{i-1}(QS_{\C})$ is closed w.r.t. 
 the set of non-generating BRs (otherwise,
 by construction of dChase, $i$ would not be a generating iteration). This implies that $c$
 is saturated before the first generating iteration $i$.
 
 (ii) for any level-1 TGC $c$, any
 BR $r \in R$, and any quad-pattern $c\colon(s,p,o)$, if $c\colon(s,p,o) \in head(r)$, then 
 for any $c'$ s.t. $c'\colon(s',p',o')$ occurs in $body(r)$ implies that $c'$ is a level-0 context (Otherwise
 level of $c$ would be greater than 1). This means that only contexts from which triples get propagated to $c$
 are level-0 contexts. From (i) we know that all the level-0 contexts are saturated before $i$th iteration, and since during the
 $i$th iteration $R_F$ is applied followed by the $\lclosure()$ operation ($R_I$ need not be applied, since $dChase_{i-1}(QS_{\C})$
is closed w.r.t. $R_I$), $c$ is saturated after iteration $i$, the 1st generating iteration.
 
 (iii) can be obtained from generalization of (i) and (ii), and from the fact that any level-$k$ context
 can only have incoming paths from contexts whose levels are less than or equal to $k$.
 $\qed$
\end{proof}
\begin{figure}[h]
\begin{subfigure}[b]{.44\textwidth}
\centering
\begin{tikzpicture}
[place/.style={circle,draw=black!50,fill=white!20,thick,
inner sep=0pt,minimum size=6mm}
]
\node(c1)[place,,label=above:$*$] at ( 0,6) [place] {$c_1$};
\node(c4)[place] at ( -2,5) [place] {$c_4$};
\node(c2) at (-3,4) [place] {$c_2$};
\node(c3)[place,,label=below:$*$] at ( -1,2) [place] {$c_3$};
\node (c5) [terminal]  at ( -5,4.5){..};
\node (c6) [terminal]  at ( -4.5,4){..};
\node (c7) [terminal]  at ( -3,1){..};
\draw [->] (c4.west) to[out=202.5,in=67.5] node[auto]{} (c2.north);
\draw [->] (c2.east) to[out=22.5,in=247.5] node[auto,swap]{} (c4.south);
\draw [->] (c3.south) to node[auto,swap]{} (c7.east);
\draw [->] (c2.west) to node[auto,swap]{} (c5.east);
\draw [->] (c2.west) to node[auto,swap]{} (c6.east);
\draw [->] (c2.south) to node[auto,swap]{} (c3.west);
\draw [->] (c4.north) to node[auto]{} (c1.west);
\draw [->] (c1.south) to node[auto,swap]{} (c3.north);
\end{tikzpicture}
\caption{}
\label{fig:saturation}
\end{subfigure}
\begin{subfigure}[b]{.44\textwidth}
\centering
\begin{tikzpicture}
[place/.style={circle,draw=black!50,fill=white!20,thick,
inner sep=0pt,minimum size=6mm}
,
special/.style={circle,draw=black!50,fill=gray!20,thick,
inner sep=0pt,minimum size=6mm}
]
\node(c1)[special,,label=above:$*$] at ( 0,6) [place] {$c_1$};
\node(c4)[place] at ( -2,5) [special] {$c_4$};
\node(c2) at (-3,4) [special] {$c_2$};
\node(c3)[place,,label=below:$*$] at ( -1,2) [place] {$c_3$};
\node (c5) [terminal]  at ( -5,4.5) {..};
\node (c6) [terminal]  at ( -4.5,4) {..};
\node (c7) [terminal]  at ( -3,1){..};
\draw [->] (c4.west) to[out=202.5,in=67.5] node[auto]{} (c2.north);
\draw [->] (c2.east) to[out=22.5,in=247.5] node[auto,swap]{} (c4.south);
\draw [->] (c2.west) to node[auto,swap]{} (c5.east);
\draw [->] (c2.west) to node[auto,swap]{} (c6.east);
\draw [->] (c3.south) to node[auto,swap]{} (c7.east);
\draw [->] (c2.south) to node[auto,swap]{} (c3.west);
\draw [->] (c4.north) to node[auto]{} (c1.west);
\draw [->] (c1.south) to node[auto,swap]{} (c3.north);
\end{tikzpicture}
\caption{}
\label{fig:saturation1}
\end{subfigure}
\caption{}
\end{figure}
\begin{example}
 Consider the dependency graph in Fig.~\ref{fig:saturation}, where $..$ indicates part of the graph that is not under the scope
 of our discussion. The TGCs nodes $c_1$ and $c_3$ are marked with a $*$.
  It can be seen that both $c_2$ and $c_4$ are level-0 contexts, since they do not have any incoming paths from TGCs.
 Since the only incoming paths to context $c_1$ are from $c_2$ and $c_4$, which are not TGCs,  
  $c_1$ is a level-1 TGC. Context $c_3$ is a level-2 TGC, since it has an incoming path from the level-1 TGC $c_1$, and
  has no incoming path from a TGC whose level is greater than 1.  Since the level-0 contexts only have 
  incoming paths from level-0 contexts and only appear on the head part of 
 non-generating BRs, before first generating iteration, all the level-0  TGCs  becomes
saturated, as the set of non-generating BRs $R_I$ has been exhaustively applied. 
This situation is reflected in Fig.~\ref{fig:saturation1}, where the saturated nodes are shaded with gray. 
Note that after the first and second generating iterations $c_1$ and $c_3$ also become saturated, respectively.
\end{example}
The following lemma shows that for context acyclic quad-systems, there exists a finite bound on the size and
computation time of its dChase.
\begin{lemma}\label{lemma:contextAcyclicChaseComputationTime}
 For any context acyclic quad-system $QS_{\C}=\langle Q_{\C}, R\rangle$, the following holds: 
 (i) the number of dChase iterations is finite,
 (ii) size of the dChase $\|dChase(QS_{\C})\|$ $=$ $\bigO(2^{2^{\|QS_{\C}\|}})$,
 (iii) computing $dChase(QS_{\C})$ is in \textsc{2EXPTIME}, 
 (iv)  if $R$ and the set of schema triples in $Q_{\C}$ is fixed, then $\|dChase(QS_{\C})\|$ is a polynomial in $\|QS_{\C}\|$, and
computing $dChase(QS_{\C})$ is  in \textsc{PTIME}.
\end{lemma}
\begin{proof}
(i) Since $QS_{\C}$ is context-acyclic, all the contexts can be partitioned according to their levels. 
Also, the number of levels $k$ is s.t. $k\leq |\C|$. Hence, 
applying lemma ~\ref{lemma:chaseSizeIncrease}, before the $k+1$th generating iteration all the contexts becomes saturated,
and $k+1$th generating iteration do not produce any new quads, terminating the dChase computation process. 

(ii) In the dChase computation process,
since by lemma~\ref{lemma:chaseSizeIncrease}, any generating iteration and a sequence of non-generating iterations
can only increase the
dChase size exponentially in $\|R\|$,  the size of the dChase before $k+1$ th generating iteration
is $\bigO(\|dChase_0(QS_{\C})\|^{{\|R\|}^k})$, which can be written
as $\bigO(\|QS_{\C}\|^{{\|R\|}^k})$ $(\dagger)$.
As seen in (i), there can only be $|\C|$ generating iterations, and a sequence of non-generating
iterations. Hence, applying $k=|\C|$ to  $(\dagger)$, and taking into account the fact that
$|\C| \leq \|QS_{\C}\|$, the size of the dChase
$\|dChase(QS_{\C})\|=\bigO(2^{2^{\|QS_{\C}\|}})$.

(iii) Since in any dChase iteration except the final one, atleast one new quad should be produced and the final dChase
can have at most $\bigO(2^{2^{\|QS_{\C}\|}})$ quads (by ii), 
the total number of iterations are bounded by $\bigO(2^{2^{\|QS_{\C}\|}})$ $(\dagger)$.
Since from lemma ~\ref{lemma:chaseSizeIncrease}, we know that for any iteration $i$, 
computing $dChase_i(QS_{\C})$ is of the order $\bigO(\|dChase_{i-1}(QS_{\C}$ $)\|^{\|R\|})$. Since,
$\|dChase_{i-1}(QS_{\C})\|$ can at most be $\bigO(2^{2^{\|QS_{\C}\|}})$, computing
$dChase_i($ $QS_{\C})$ is of the order $\bigO(2^{\|R\|*2^{\|QS_{\C}\|}}))$. Also since $\|R\|\leq \|QS_{\C}\|$,
any iteration requires $\bigO(2^{2^{\|QS_{\C}\|}})$ time $(\ddagger)$. 
From $(\dagger)$ and $(\ddagger)$, we can conclude that the time required for computing dChase is in 2EXPTIME.

(iv) In (ii) we saw that the size of the dChase before $k+1$th generating iteration is
given by $\bigO(\|QS_{\C}\|^{{\|R\|}^k})$ $(\diamond)$. Since by hypothesis $\|R\|$ is a constant and
also the size of the dependency graph and the levels in it. Hence, the expression  
${{\|R\|}^k}$ in $(\diamond)$ amounts to a constant $z$. Hence, 
$\|dChase(QS_{\C})\|$ $=$ $\bigO(\|QS_{\C}\|^z)$. Hence, the size of $dChase(QS_{\C})$ is a polynomial in $\|QS_{\C}\|$. 

Also, since in any dChase iteration except the final one, atleast one quad should be produced and the final dChase
can have at most $\bigO(\|QS_{\C}\|^z)$ quads, the total number of iterations are bounded by $\bigO(\|QS_{\C}\|^z)$ $(\dagger)$.
Also from lemma ~\ref{lemma:chaseSizeIncrease}, we know that any dChase iteration $i$, 
computing $dChase_i(QS_{\C})$ involves two steps: (a) computing $R(dChase_{i-1}(QS_{\C}))$, and (b) computing $\lclosure()$, which can be done in 
PTIME in the size of its input. 
Since computing $R(dChase_{i-1}(QS_{\C}))$ is of the order $\bigO(\|dChase_{i-1}(QS_{\C})\|^{\|R\|})$, where
$|R|$ is a constant and $\|dChase_{i-1}(QS_{\C})\|$ is a polynomial is $\|QS_{\C}\|$,
each iteration can be done in time polynomial in $\|QS_{\C}\|$ ($\ddagger$).
From $(\dagger)$ and $(\ddagger)$, it can be concluded that dChase can be computed in PTIME. $\qed$
\end{proof}

\begin{lemma}\label{lemma:context-acyclic-computational-properties}
For any context acyclic quad-system, the following holds: 
 (i) data complexity of CCQ entailment is in \textsc{PTIME}
 (ii) combined complexity of CCQ entailment is in \textsc{2EXPTIME}. 
\end{lemma}
\begin{proof}
For a context acyclic quad-system $QS_{\C}=\langle Q_{\C}, R\rangle$, since $dChase(QS_{\C})$ is finite, 
a boolean CCQ $CQ()$ can naively be evaluated by grounding the set of constants in the chase to the variables in the $CQ()$,
and then checking if any of these groundings are contained in $dChase(QS_{\C})$. The number of such 
groundings can at most be $\|dChase(QS_{\C})\|^{\|CQ()\|}$ ($\dagger$).

 (i) Since for data complexity, the size of the BRs $\|R\|$, the set of schema triples, and $\|CQ()\|$ is fixed to constant. 
 From lemma \ref{lemma:contextAcyclicChaseComputationTime} (iv), we know that under the above mentioned settings the dChase can be 
 computed in PTIME and is polynomial in the size of $QS_{\C}$. Since $\|CQ()\|$ is fixed to a constant, and 
 from ($\dagger$), binding the set of constants in $dChase(QS_{\C})$ on $CQ()$ still gives a number of bindings that is
 worst case polynomial in the size of $QS_{\C}$. Since membership of these bindings can checked in the polynomially sized 
 dChase in PTIME, the time required for CCQ evaluation is in PTIME.
 
 (ii) Since in this case $\|dChase(QS_{\C})\|=\bigO(2^{2^{\|QS_{\C}\|}})$ $(\ddagger)$, from ($\dagger$) and $(\ddagger)$,  
 binding the set of constants in $\|dChase(QS_{\C}\|$ to variables in $CQ()$ amounts to $\bigO(2^{\|CQ()\|*2^{\|QS_{\C}\|}})$ 
 bindings. Since the size of dChase is double exponential in $\|QS_{\C}\|$, checking the membership of each of these bindings
 can be done in 2EXPTIME. Hence, the combined complexity is in 2EXPTIME.
 $\qed$
\end{proof}
\begin{theorem}\label{thoeorem:context-acyclic-computational-properties}
 For any context acyclic quad-system, the following holds: 
 (i) The data complexity of CCQ entailment is \textsc{PTIME}-complete, (ii) 
The combined complexity of CCQ entailment is \textsc{2EXPTIME}-complete. 
\end{theorem}
For PTIME-hardness of data complexity, it can be shown that the well known
problem of 3HornSat, the satisfiability of propositional Horn formulas with atmost 3 literals, and for 2EXPTIME-hardness for the combined complexity, it can be shown that
the word problem of a double exponentially time bounded deterministic turing machine, which is a well known 2EXPTIME-hard problem,
is reducible to the CCQ entailment problem (see appendix for detailed proof). 

Reconsidering the quad-system in example~\ref{eg:dependency}, which is not context acyclic.
Suppose that the contexts are enabled with RDFS inferencing, i.e $\lclosure()=\textsf{rdfsclosure}()$. 
 During dChase construction, since any application of rule~(\ref{eqn:BRinExample}) can only create
a triple in $c_2$ in which the skolem blank node is in the object position, where as the application of rule~(\ref{eqn:BRinExample2}),
does not propogate constants in object postion to $c_1$.  
Although at a first look, the dChase might seem to
terminate, but since the application of the following RDFS inference rule in $c_2$:
$ (s,p,o)$ $\rightarrow$ $(o$ ,\texttt{rdf:type}, \texttt{rdfs:Resource}),
derives a quad of the form $c_2\colon(\_\text{ :}b$, \texttt{rdf:type}, \texttt{rdfs:Resource}), where $\_\text{ :}b$ is the skolem blank-node
created by the application of rule~(\ref{eqn:BRinExample}). Now by application of rule (\ref{eqn:BRinExample2}) leads to
$c_1\colon(\_\text{ :}b,\texttt{rdf:type}, U_1)$.
Since rule~(\ref{eqn:BRinExample}) is applicable on $c_1\colon(\_\text{ :}b,\texttt{rdf:type}, U_1)$, 
which again brings a new skolem blank node to $c_2$, and hence the dChase construction doesn't terminate.
Hence, as seen above the notion of context acyclicity can alarm us about such infinite cases.

\vspace{-10pt}
\section{Related Work}\label{sec:related work}
\paragraph{Contexts and Distributed Logics} The work on contexts began in the 80s when 
McCarthy~\cite{mccarthy_generality_87} proposed context as a 
solution to the generality problem in AI. 
After this various studies about logics of contexts mainly in the field of KR was done by Guha ~\cite{guha's_thesis}, 
\emph{Distributed First Order Logics} by Ghidini et al.~\cite{Ghidini98distributedfirst} and 
\emph{Local Model Semantics} by Giunchiglia et al.~\cite{Ghidini01localmodels}. 
Primarily in these works contexts are formalized as a first order/propositional
theories, and bridge rules were provided to inter-operate the various contexts. Some of the initial works on contexts
relevant to semantic web were the ones like \emph{Distributed Description Logics}~\cite{DDL} by Borgida et al., 
\emph{E-connections}~\cite{econnections} by Kutz et al.,
 \emph{Context-OWL}~\cite{cowl} by Bouqet et al., and the  work of 
 CKR~\cite{serafini2012contextualized,josephWomo} by Serafini et al. 
These were mainly logics based on DLs, which formalized contexts as OWL KBs, whose semantics is given using a 
distributed interpretation
structure with additional semantic conditions that suits varying requirements. Compared to these works, 
the bridge rules we consider are much more expressive with conjunctions and existential variables that supports 
value/blank-node creation. 

\paragraph{$\forall \exists$ rules, TGDs, Datalog+- rules}
Query answering over rules with universal existential quantifiers
in the context of databases/KR, where these rules are called tuple generating dependencies (TGDs)/Datalog+- rules,
was done by Beeri and Vardi~\cite{BeeriVImplicationProblem81} even in the early 80s, where the authors
show that the query entailment problem in general is undecidable. However, recently many classes of such rules 
have been identified for which query answering is decidable. 
These includes (a) fragments s.t. the resulting models 
have  bounded tree widths, called bounded treewidth sets (BTS), such as Weakly
guarded rules~\cite{Datalog+-}, Frontier guarded rules~\cite{BagetMRT11}, 
(b) fragments called finite unification sets (FUS), such as `sticky' 
rules~\cite{DBLP:conf/rr/CaliGP10,geottlobMannaPieris2014}, and (c) fragments called finite extension sets (FES), 
where sufficient conditions are enforced to ensure finiteness of the chase and its termination. 
The approach used for query answering in FUS is to rewrite the input query w.r.t. to the TGDs
to another query that can be evaluated directly on the set of instances, s.t. the answers for the former query
and latter query coincides. The approach is called the \emph{query rewriting approach}.  
FES classes uses certain termination guarantying tests that check whether certain
sufficient conditions are satisfied by the structure of TGDs.
A large number of classes in FES are based on tests that detects `acyclicity conditions' 
by analyzing the information flow between the TGD rules. \emph{Weak acyclicity}~\cite{Fagin05dataexchange,Deutsch03reformulationof},
was one of the first such notions, and was extended to \emph{joint acyclicity}~\cite{KR11jointacyc},  
\emph{super weak acyclicity}~\cite{Marnette:2009}, and \emph{model faithful acyclicity}~\cite{Bernardo:dlacyclicity:2012}. 
The most similar approach to ours is the weak acyclicity technique, where the 
structure of the rules is analyzed using a dependency graph that models the propagation of constants across various predicates positions, 
and restricting the dependency graph to be acyclic. Although this technique can be used in our scenario by translating a quad-system
to a set of TGDs; if the obtained translation is weakly acyclic, then one could use existing algorithms for chase computation
for the TGDs to compute the chase, the query entailment check can be done by querying the obtained chase.  
However, our approach has the advantage of straightforward implementability on existing quad-stores.

\vspace{-2pt}
\section{Summary and Conclusion}\label{sec:conclusion}
\begin{table*}[t]
\centering
\begin{tabular}{| c | c | c | c |}
\hline
&&&\\
Quad-System & dChase size w.r.t &  Data Complexity of \ & Combined Complexity \\
 Fragment & input quad-system \ & CCQ entailment & of CCQ entailment   \\
\hline
Unrestricted Quad-Systems & Infinite & Undecidable & Undecidable\\ 
Context acylic Quad-Systems & Double exponential  & PTIME-complete  & 2EXPTIME-complete\\ 
\hline
\end{tabular}
\caption{Complexity info for various quad-system fragments} 
\label{tab:compResults}
\end{table*}
\noindent 
In this paper, we study the problem of query answering over contextualized RDF knowledge. We show that the problem
in general is undecidable, and present a decidable class called context acyclic quad-systems. 
Table \ref{tab:compResults} summarizes the main results obtained. 
We can show that the notion of context acyclicity, introduced in section \ref{sec:cAcyclic}
can be used to extend the currently established tools  
for contextual reasoning 
to give support for expressive
BRs with conjuction and existentials with decidability guarantees.
We view the results obtained in this paper as a general foundation 
for contextual reasoning and query answering over contextualized RDF knowledge formats such as Quads, 
and can straightforwardly be used to extend existing Quad stores 
to encorporate for-all existential BRs of the form (\ref{eqn:intraContextualRuleWithQuantifiers}).

\bibliographystyle{splncs03} 
\bibliography{paper}
\appendix
\section{Proofs of Section \ref{sec:Query}}

\begin{proof}[Theorem \ref{theorem:undecidable}]
We show that CCQ entailment is undecidable for
unrestricted quad-systems, by showing that the well known undecidable problem of 
``non-emptiness of intersection of context-free grammars'' is reducible to the CCQ entailment
problem. 

Given an alphabet $\Sigma$,  string $\vec w$ is a sequence of symbols from $\Sigma$. A language $L$ is a 
subset of $\Sigma^*$, where $\Sigma^*$ is the set of all strings that can be constructed
from the alphabet $\Sigma$, and also includes the empty string $\epsilon$. Grammars are machineries that generate a particular language.
A grammar $G$ is a quadruple $\langle V, T, S, P \rangle$, where $V$ is the set of variables, $T$, the set of terminals,
$S \in V$ is the start symbol, and $P$ is a set of production rules (PR), in which each PR $r \in P$, is of the form:
\[
\vec w \rightarrow \vec w'
\]
where $\vec w, \vec w' \in \{T \cup V\}^*$. Intuitively application of a PR $r$ of the form above on a string $ \vec w_1$, 
replaces every occurrence of the sequence $\vec w$ in $\vec w_1$ with $\vec w'$. PRs are applied starting from the 
start symbol $S$ until it results in a string $\vec w$, with $\vec w \in \Sigma^*$ or no more production 
rules can be applied on $\vec w$. In the former case, we say that $\vec w \in L(G)$, the language generated by grammar $G$. For
a detailed review of grammars, we refer the reader to Harrison et al. ~\cite{Harrison-formal-language}.
A \emph{context-free grammar} (CFG) is a grammar, whose set of PRs $P$, have the following property: 
\begin{property}\label{prop:CFG}
For a CFG, every PR is of the form $v \rightarrow \vec w$, where $v \in V$, $\vec w \in \{T \cup V\}^*$.
\end{property}
 Given two CFGs, $G_1=\langle V_1, T, S_1, P_1 \rangle$ and $G_2=\langle V_2, T, S_2, P_2 \rangle$, 
where $V_1, V_2$ are the set of variables, 
$T$ s.t. $T \cap (V_1 \cup V_2)=\emptyset$ is the set of terminals. $S_1 \in V_1$ is the start symbol of $G_1$,
and $P_1$ are the set of PRs of the form $v \rightarrow \vec w$, where $v \in V$, $\vec w$ is a sequence of the form
$w_1...w_n$, where $w_i \in V_1 \cup T$. Similarly $s_2, P_2$ is defined.  Deciding whether the language
generated by the grammars $L(G_1)$ and $L(G_2)$ have non-empty intersection is known 
to be undecidable~\cite{Harrison-formal-language}.

Given two CFGs, $G_1=\langle V_1, T, S_1, P_1 \rangle$ and $G_2=\langle V_2, T, S_2, P_2 \rangle$,
we encode grammars $G_1, G_2$ into a quad-system of the form $QS_c=\langle Q_{c},R\rangle$, with a single context identifier $c$.
Each PR $r= v \rightarrow \vec w \in P_1 \cup P_2$, with $\vec w=w_1w_2w_3..w_n$, is encoded as a BR of the form:
\begin{eqnarray}\label{eqn:PRtoICR}
c\colon(x_1,w_1,x_2), c\colon(x_2, w_2, x_3),...,c\colon(x_n,w_n,x_{n+1}) 
\rightarrow c\colon(x_1,v,x_{n+1}) 
\end{eqnarray}
where $x_1,..,x_{n+1}$ are variables. W.l.o.g. we assume that the set of terminal symbols $T$ is
equal to the set of terminal symbols occurring in $P_1 \cup P_2$. For each terminal symbol $t_i \in T$, $R$ contains
a BR of the form:
\begin{eqnarray}\label{eqn:existential rule}
 c\colon(x, \texttt{rdf:type}, C) \rightarrow \exists y \ c\colon(x,t_i,y), 
 c\colon(y, \texttt{rdf:type},C) 
\end{eqnarray}
and $Q_c$ contains only the triple:
\[
 c\colon(a,\texttt{rdf:type},C)
\]
We in the following show that:
\[
QS_c \models \exists y \ c\colon(a , S_1, y) \wedge c\colon(a, S_2, y) \Leftrightarrow L(G_1)\cap L(G_2)\neq\emptyset
\]
\begin{claim}(1)
 For any $\vec w=t_1,...,t_p \in T^*$, there exists $b_1,...b_p$, s.t. $c\colon(a,t_1,b_1)$, $c\colon(b_1,t_2,b_2)$, ..., 
 $c\colon(b_{p-1}, t_p, b_p)$, $c\colon(b_p,\texttt{rdf:type},C)$ $\in$ $dChase($ $QS_c)$.
\end{claim}
we proceed by induction on the $\|\vec w\|$.
\begin{description}
 \item [base case] suppose if $\|\vec w\|=1$, then $\vec w=t_i$, for some $t_i \in T$. But
 Since by construction $c\colon(a$, \texttt{rdf:type}, $C)$ $\in$ $dChase_0(QS_c)$, on which rules of the 
 form (\ref{eqn:existential rule}) is applicable. Hence, there exists an $i$ s.t.
 $dChase_i(QS_c)$ contains $c\colon(a,t_i,b_i)$, $c\colon(b_i,\texttt{rdf:type},C)$, for each $t_i \in T$. 
 Hence, the base case.
 \item [hypothesis] for any $\vec w=t_1...t_p$, if $\|\vec w\|\leq p'$, then there exists $b_1,...,b_p$, s.t. 
 $c\colon(a,t_1,b_1)$, $c\colon(b_1,t_2,b_2)$, ..., $c\colon(b_{p-1}, t_p, b_p)$, $c\colon(b_p$, \\
 \texttt{rdf:type}, $C)$ $\in$ $dChase(QS_c)$.
 \item [inductive step] suppose $\vec w=t_1...t_{p+1}$, with $\|\vec w\|\leq p'+1$. Since $\vec w$ can be written
 as $\vec{w'}t_{p+1}$, where $\vec w'=t_1...t_p$, and by hypothesis, there exists $b_1,...,b_p$ such 
 that  $c\colon(a,t_1,b_1)$, $c\colon(b_1,t_2,b_2)$, $...$, $c\colon(b_{p-1}, t_p, b_p)$, $c\colon(b_p,\texttt{rdf:type},C)$ $\in$ 
 $dChase(QS_c)$. Also since rules of the form (\ref{eqn:existential rule}) are applicable on $c\colon(b_p$, \texttt{rdf:type}, $C)$,
 and hence produces triples of the form $c\colon(b_p,t_i,b_{p+1}^i)$, $c\colon(b_{p+1}^i)$, \texttt{rdf:type}, $C)$, 
 for each $t_i \in T$. Since $t_{p+1} \in T$, the claim follows.
\end{description}
 For a grammar $G=\langle V,T,S,P\rangle$, whose start symbol is $S$, and for any $\vec w \in \{V \cup T\}^*$,
 for some $V_j \in V$, we denote by $V_j \rightarrow^i \vec w$, 
 the fact that $\vec w$ was derived from $V_j$ by $i$ production steps, i.e. there exists steps 
 $V_j \rightarrow r_1, ...,r_i \rightarrow \vec w$, which lead to the production of $\vec w$. 
For any $\vec w$, $\vec w \in L(G)$, iff 
there exists an $i$ s.t. $S \rightarrow^i \vec w$. For any $V_j \in V$, we use $V_j \rightarrow^* \vec w$ to denote
the fact that there exists an arbitrary $i$, s.t. $V_j \rightarrow^i \vec w$.
\begin{claim}(2)
 For any $\vec w=t_1...t_p \in \{V \cup T\}^*$, and for any $V_j \in V$, if $V_j \rightarrow^* \vec w$
and there exists $b_1,...,b_{p+1}$, with $c\colon(b_1,t_1,b_2), ..., c\colon(b_p,t_p,b_{p+1}) \in dChase(QS_c)$,
then $c\colon(b_1,V_j,b_{p+1}) \in dChase(QS_c)$. 
\end{claim}
We prove this by induction on the size of $\vec w$.
\begin{description}
 \item[base case] Suppose $\|\vec w\|=1$, then $\vec w=t_k$, for some $t_k \in T$. If there exists
 $b_1,b_2$ s.t. $c\colon(b_1,t_k,b_2)$. But since there exists a PR $V_j \rightarrow t_k$,  by 
 transformation given in (\ref{eqn:PRtoICR}), there exists a BR $c\colon(x_1,t_k,x_2) \rightarrow c\colon(x_1,V_j,x_2)
 \in R$, which is applicable on $c\colon(b_1,t_k,b_2)$ and hence the quad $c\colon(b_1,V_j,b_2) \in dChase(QS_c)$.
 \item[hypothesis] For any $\vec w=t_1...t_p$, with 
 $\|\vec w\|\leq p'$, and for any $V_j \in V$, if $V_j \rightarrow^* \vec w$
and there exists $b_1,...b_p,b_{p+1}$, s.t. $c\colon(b_1,t_1,b_2)$, $...$, $c\colon(b_p,t_p,b_{p+1})$ $\in$ $dChase(QS_c)$,
then $c\colon(b_1$, $V_j$, $b_{p+1})$ $\in$ $dChase(QS_c)$.
 \item[inductive step] Suppose if $\vec w=t_1...t_{p+1}$, with $\|\vec w\|\leq p'+1$, and  $V_j \rightarrow^i \vec w$, 
and there exists $b_1,...b_{p+1}$, $b_{p+2}$, s.t. $c\colon(b_1,t_1,b_2)$, $...$, 
$c\colon(b_{p+1}$, $t_{p+1}$, $b_{p+2})$ 
 $\in$ $dChase(Q_c)$.  Also, one of the following holds (i) $i=1$, or (ii) $i>1$. Suppose (i) is the case, 
 then it is trivially the case that $c\colon(b_1,V_j,b_{p+2}) \in dChase(QS_c)$. Suppose if (ii) is the case,
 one of the two sub cases holds (a) $V_j \rightarrow^{i-1} V_k$, for some $V_k \in V$ and $V_k \rightarrow^1 \vec w$ or (b)
   there exist a $V_k \in V$, s.t. $V_k \rightarrow^*  t_{q+1}...t_{q+l}$, with $2\leq l \leq p$, where
 $V_j \rightarrow^* t_1...t_qV_kt_{p-l+1}...t_{p+1}$. If (a) is the case, trivially
 then $c\colon(b_1,V_k,b_{q+2}) \in dChase(QS_c)$, and since by construction there exists $c\colon(x_0, V_k, x_1)$ $\rightarrow$ 
 $c\colon(x_0,V_{k+1},x_1)$, $...$, $c\colon(x_0,V_{k+i},x_1)$ $\rightarrow$ $c\colon(x_0,V_j,x_1)$ $\in$ $R$,
$c\colon(b_1,V_j,b_{q+2}) \in dChase($ $QS_c)$.
If (b) is the case, then since $\|t_{q+1}...t_{q+l}\|\geq 2$, 
 $\|t_1...t_qV_2t_{p-l+1}...t_{p+1}\|\leq p'$. This implies that  
 $c\colon(b_1,V_j,b_{p+2}) \in dChase(QS_c)$.
\end{description}
Similarly, by construction of $dChase(QS_c)$, the following claim can straightforwardly be shown to hold:
\begin{claim}(3)
 For any $\vec w=t_1...t_p \in \{V \cup T\}^*$, and for any $V_j \in V$, if there exists 
 $b_1,...,b_p,b_{p+1}$, with $c\colon(b_1,t_1,b_2)$, ..., $c\colon(b_p,t_p,b_{p+1})$ $\in$ $dChase(QS_c)$ 
 and $c\colon(b_1$, $V_j$, $b_{p+1})$ $\in$ $dChase(QS_c)$, then $V_j \rightarrow^* \vec w$. 
\end{claim}
(a) For any $\vec w=t_1...t_p \in T^*$, if $\vec w \in L(G_1) \cap L(G_2)$, then
by claim 1, since there exists $b_1,...,b_p$, s.t. $c\colon(a,t_1,b_1),...,c\colon(b_{p-1},t_p,b_p) \in dChase(QS_c)$. But
since $\vec w \in L(G_1)$ and $\vec w \in L(G_2)$, $S_1 \rightarrow \vec w$ and $S_2 \rightarrow \vec w$. Hence 
by claim 2, $c\colon(a,S_1,b_p),c\colon(a,S_2,b_p)$ $\in$ $dChase(QS_c)$, which implies that $dChase(QS_c)$ $\models$ 
$\exists y$  $c\colon(a,s_1,y)$ $\wedge$ $c\colon(a,s_2,y)$. Hence, $QS_c$ $\models$  
$\exists y$ $c\colon(a,s_1,y)$ $\wedge$ $c\colon(a,s_2,y)$.  \\
(b) Suppose if $QS_c \models  \exists y \ c\colon(a,S_1,y) \wedge c\colon(a,S_2,y)$, then this implies that 
there exists $b_p$ s.t. $c\colon(a$, $S_1$, $b_p)$, $c\colon(a,S_2,b_p) \in dChase(QS_{\C})$. 
Then it is the case that there exists $\vec w=t_1...t_p \in T^*$, and
$b_1,...,b_p$ s.t. $c\colon(a,t_1,b_1)$, ..., $c\colon(b_{p-1},t_p,b_p)$, $c\colon(a,S_1,b_p)$, 
$c\colon(a,S_2,b_p)$ $\in$ $dChase(QS_c)$. Then by claim 3,
$S_1 \rightarrow^* \vec w$, $S_2 \rightarrow^* \vec w$. Hence, $w \in L(G_1) \cap L(G_2)$. 

By (a),(b) it follows that there exists $\vec w \in L(G_1) \cap L(G_2)$ 
iff $QS_c \models \exists y \ c\colon(a,s_1,y) \wedge c\colon(a,s_2,y)$.
As we have shown that the intersection of CFGs, which is an undecidable problem, is reducible to the problem of
query entailment on unrestricted quad-system, the latter is undecidable.
\end{proof}

\section{Proofs of Section \ref{sec:cAcyclic}}\label{AppendixSec:csafe}

\begin{proof}[Theorem \ref{thoeorem:context-acyclic-computational-properties}]
 
 (i) (Membership) See lemma \ref{lemma:context-acyclic-computational-properties} for the membership in PTIME.
 
 (Hardness)
In order to prove P-hardness, we reduce a well known P-complete problem, 3HornSat, i.e. the satisfiability of 
propositional Horn formulas with at most 3 literals. Note that a (propositional) Horn formula is formula of the form:
\begin{eqnarray}\label{eqn:PropHornFormula}
 P_1 \wedge \ldots \wedge P_n \rightarrow P_{n+1}  
\end{eqnarray}
where $P_i$, for  $1 \leq i \leq n+1$, are either propositional variables or constants $t$, $f$, that represents
true and false, respectively. Note that for any propositional variable $P$, the fact that ``$P$ holds'' is 
represented by the formula $t \rightarrow P$, and ``$P$ does not hold'' is 
represented by the formula $P \rightarrow f$.
A 3Horn formula is a formula of the form~(\ref{eqn:PropHornFormula}), where $1\leq n\leq 2$. Note that any (set of) 
Horn formula(s) $\Phi$ can be transformed in polynomial time to a polynomially sized set $\Phi'$ of 3Horn formulas,
by introducing auxiliary propositional variables s.t. $\Phi$ is satisfiable iff $\Phi'$ is satisfiable.
A pure 3Horn formula is 3Horn formula of the form~\ref{eqn:PropHornFormula}, where $n=2$. Any 3Horn formula $\phi$ that is not pure
can be trivially converted to equivalent pure form by appending a $\wedge \ t$ on the head part of $\phi$. For instance, 
 $P\rightarrow Q$,  can be converted to  $P \wedge t \rightarrow Q$. Hence, w.l.o.g. we assume that any set of 3Horn formulas is pure, and is of the form:
 \begin{eqnarray}\label{eqn:PropPure3HornFormula}
  P_1 \wedge P_2 \rightarrow P_3  
 \end{eqnarray}
 We, in the following, reduce the satisfiability problem of pure 3Horn formulas to CCQ 
entailment problem over a quad-system whose set of schema triples, the set of BRs, and the CCQ $CQ$ are all fixed. 

 For any set of pure Horn formulas $\Phi$, we construct the quad-system $QS_{\C}=\langle Q_{\C}, R\rangle$, where
 $\C=\{c_t, c_f\}$. For any formula $\phi \in \Phi$ of the form (\ref{eqn:PropPure3HornFormula}), $Q_{\C}$ contains
 a quad $c_f \colon (P_1, P_2, P_3)$. In addition $Q_{\C}$ contains a quad $c_t\colon (t$, \texttt{rdf:type}, $T)$. 
 $R$ is the singleton that contains only the following fixed BR:
 \begin{eqnarray}
  c_t\colon (x_1, \texttt{rdf:type}, T), c_t\colon (x_2,  \texttt{rdf:type}, T), c_f\colon (x_1, x_2, x_3) \rightarrow c_t\colon (x_3,  \nonumber \\
  \texttt{rdf:type}, T) \nonumber 
 \end{eqnarray}
Let the $CQ$ be the fixed query $c_t(f, \texttt{rdf:type}, T)$. 

Now, it is easy to see that $QS_{\C}$ $\models$ $CQ$, iff $\Phi$ is not satisfiable. 

 (ii) 
 (Membership) See lemma \ref{lemma:context-acyclic-computational-properties}.
 
 (Hardness) See following heading.
 $\qed$
\end{proof}

\subsubsection{2EXPTIME-Hardness of CCQ Entailment}
In this subsection, we show that the combined complexity of the decision problem of CCQ entailment for context 
acyclic quad-systems is 2EXPTIME-hard. We show this
by reduction of the word-problem of a 2EXPTIME deterministic turing machine (DTM) 
to the CCQ entailment problem. A DTM $M$ is a tuple $M=\langle Q, \Sigma, \Delta, q_0, q_A\rangle$, where 
\begin{itemize}
 \item $Q$ is a set of states,
 \item $\Sigma$ is a finite alphabet that includes the blank symbol $\Box$,
 \item $\Delta\colon (Q \times \Sigma) \rightarrow (Q \times \Sigma \times \{+1, -1\})$ is the transition function,
 \item $q_0 \in Q$ is the initial state.
 \item $q_A \in Q$ is the accepting state.
\end{itemize}
 W.l.o.g. we assume that there exists exactly one accepting state, which is also a halting state. 
A  configuration is a word $\vec \alpha \in \Sigma^*Q\Sigma^*$. A
configuration $\vec \alpha_2$ is a successor of the configuration $\vec \alpha_1$, iff one of the following holds:
\begin{enumerate}
 \item $\vec \alpha_1=\vec w_lq\sigma\sigma_r \vec w_r$ and $\vec \alpha_2=\vec w_l\sigma'q'\sigma_r \vec w_r$, 
 if $\Delta(q,\sigma) =(q',\sigma',R) $, or
  \item $\vec \alpha_1=\vec w_lq\sigma$ and $\vec \alpha_2=\vec w_l\sigma'q'\Box$, 
  if $ \Delta(q,\sigma)=(q',\sigma',R)$, or
  \item $\vec \alpha_1=\vec w_l\sigma_lq\sigma \vec w_r$ and $\vec \alpha_2=\vec w_lq'\sigma_l\sigma' \vec w_r$, 
  if $ \Delta(q,\sigma)=(q',\sigma',L)$.
\end{enumerate}
where $q,q' \in Q$, $\sigma,\sigma',\sigma_l, \sigma_r \in \Sigma$, and $\vec w_l, \vec w_r\in \Sigma^*$. 
Since number of configurations can at most be doubly exponential in the size of the input string,  
the number of tape cells traversed by the DTM tape head is also bounded double exponentially. A configuration $\vec c=\vec w_lq\vec w_r$ is an accepting configuration iff
$q=q_A$. A language $L \subseteq \Sigma^*$ is accepted by a 2EXPTIME bounded DTM $M$, iff 
for every $\vec w \in L$, $M$ accepts $\vec w$ in time $\bigO(2^{2^{\|\vec w\|}})$. 

\subsubsection{Simulating DTMs using Context Acyclic Quad-Systems}\label{subsection:ATM simulation}
 Consider a DTM $M=\langle Q,\Sigma, \Delta, q_0, q_A\rangle$, and a string $\vec w$, with $\|\vec w\|=n$. Since the 
number of storage cells is double exponentially bounded, we first construct a quad-system $QS^M_{\C}=\langle Q^M_{\C}, R\rangle$,
where  $\C=\{c_1,...,c_n\}$, with $n=\|\vec w\|$. We follow the technique in works such as 
\cite{DBLP:conf/rr/CaliGP10,KR11jointacyc} to iteratively 
generate a doubly exponential number of objects that represent the cells of the tape of the DTM. 
Let $Q^M_{\C}$
be initialized with the following quads:
\begin{eqnarray}
&& c_0\colon(k_0,\texttt{rdf:type},R),c_0\colon(k_1,\texttt{rdf:type},R), \nonumber \\
&& c_0\colon(k_0,\texttt{rdf:type}, min_0), c_0\colon(k_1,\texttt{rdf:type},  
 max_0),  c_0\colon(k_0,succ_0,k_1) \nonumber
\end{eqnarray}
Now for each pair of elements of type $R$ in $c_i$, a skolem blank-node is generated
in $c_{i+1}$, and hence follows the recurrence relation $r(m+1)=[r(m)]^2$, with seed $r(1)=2$, which after $n$ iterations
yields $2^{2^n}$. In this way, a doubly exponential long chain of elements is created in 
$c_n$ using the following set of rules:
\begin{eqnarray}
&& c_i\colon(x_0, \texttt{rdf:type},R),c_i\colon(x_1,\texttt{rdf:type},R) \rightarrow \nonumber \\
&& \exists y \ c_{i+1}\colon(x_0,x_1,y), 
c_{i+1}\colon(y,\texttt{rdf:type},R) \nonumber 
\end{eqnarray}
The combination of minimal element with the minimal element (elements of type $min_i$) in $c_i$
create the minimal element in $c_{i+1}$, and similarly the combination of maximal element with the 
maximal element (elements of type $max_i$) in $c_i$
create the maximal element of $c_{i+1}$
\begin{eqnarray}
c_{i+1}\colon(x_0,x_0,x_1), c_i\colon(x_0,\texttt{rdf:type},min_i) \rightarrow 
c_{i+1}\colon(x_1,\texttt{rdf:type},min_{i+1}) \nonumber \\
c_{i+1}\colon(x_0,x_0,x_1), c_i\colon(x_0,\texttt{rdf:type},max_i) \rightarrow 
c_{i+1}\colon(x_1,\texttt{rdf:type},max_{i+1}) \nonumber 
\end{eqnarray}
Successor relation $succ_{i+1}$ is created in $c_{i+1}$ using the following set of rules, using the well-known,
integer counting technique:
\begin{eqnarray}
&& c_i\colon(x_1,succ_i,x_2), c_{i+1}\colon(x_0,x_1,x_3), 
 c_{i+1}\colon(x_0,x_2,x_4) \rightarrow c_{i+1}\colon(x_3, succ_{i+1},x_4) \nonumber \\
&& \nonumber \\
&& c_i\colon(x_1,succ_i,x_2), c_{i+1}\colon(x_1,x_3,x_5),c_{i+1}\colon(x_2,x_4,x_6), c_i\colon(x_3, \texttt{rdf:type},max_i), \nonumber \\
&&  c_i\colon(x_4, \texttt{rdf:type},  
   min_i) \rightarrow c_{i+1}\colon(x_5, succ_{i+1},x_6) \nonumber 
\end{eqnarray}
Each of the above set rules are instantiated for $0\leq i <n$, and in this way after $n$ generating dChase iterations, 
$c_n$ has doubly exponential number of elements of type $R$, that are ordered linearly using the relation $succ_n$.
By virtue of the first rule below, each of the objects representing the cells of the DTM are 
linearly ordered by the relation $succ$. Also the transitive closure of $succ$ is defined as the relation
relation $succt$
\begin{eqnarray}
&& c_n\colon(x_0, succ_n, x_1) \rightarrow c_n\colon(x_0, succ, x_1) \nonumber \\
&& c_n\colon(x_0, succ, x_1) \rightarrow c_n\colon(x_0, succt,x_1) \nonumber \\
&& c_n\colon(x_0, succt, x_1), c_n\colon(x_1,succt,x_2) 
 \rightarrow c_n\colon(x_0, succt, x_2) \nonumber
\end{eqnarray}
Also using a similar construction, we could create a linearly ordered chain 
of double exponential number of objects in $c_n$ that 
represents configurations of $M$, whose minimal element is of type $conInit$, and the linear order relation being $conSucc$.

 Various triple patterns that are used to encode the possible configurations, runs and their relations in $M$ are:
\begin{description}
 \item[$(x_0,head,x_1)$] denotes the fact that in configuration $x_0$, the head of the DTM is at cell $x_1$.
 \item[$(x_0,state, x_1)$] denotes the fact that in configuration $x_0$, the DTM is in state $x_1$.
 \item[$(x_0,\sigma, x_1)$] where $\sigma \in \Sigma$, denotes the fact that in configuration
 $x_0$, the cell $x_1$ contains $\sigma$.
 \item[$(x_0,succ, x_1)$] denotes the linear order between cells of the tape.
 \item[$(x_0,succt, x_1)$] denotes the transitive closure of $succ$.
 \item[$(x_0,conSucc,x_1)$] to denote the fact that $x_1$ is a successor configuration of $x_0$.
 \item$(x_0,\texttt{rdf:type},Accept)$ denotes the fact that the configuration $x_0$ is an accepting configuration.
\end{description}
Since in our construction, each $\sigma \in \Sigma$ is represented as relation, we could constrain
that no two alphabets $\sigma \neq \sigma'$ are on the same cell  using the 
following axiom:
\begin{eqnarray}
  c_n\colon(z_1 ,\sigma, z_2), c_n\colon(z_1, \sigma', z_2) \rightarrow  \nonumber
\end{eqnarray}
for each $\sigma \neq \sigma' \in \Sigma$. Note that the above BR has an empty head, is equivalent to asserting the negation of its body. 

\paragraph{Initialization}
Suppose the initial configuration is $q_0\vec w\Box$, where $\vec w=\sigma_0...\sigma_{n-1}$, then we 
enforce this using the following BRs in our quad-system $QS^M_{\C}$ as:
\begin{eqnarray}
&& c_n\colon(x_0,\texttt{rdf:type}, conInit), c_n\colon(x_1,\texttt{rdf:type}, min_n) \rightarrow c_n\colon(x_0,head,x_1),\nonumber \\ 
&& c_n\colon(x_0,state,q_0) \nonumber \\ 
&& c_n\colon(x_0,\texttt{rdf:type}, min_n) \wedge \bigwedge_{i=0}^{n-1} c_n\colon(x_i, succ,x_{i+1}) \wedge c_n\colon(x_j, \texttt{rdf:type}, \nonumber \\
&& conInit)  \rightarrow \bigwedge_{i=0}^{n-1} c_n\colon(x_j, \sigma_i, x_i) \wedge c_n\colon(x_j,\Box, x_n) \nonumber \\
&& c_n\colon (x_j, \texttt{rdf:type}, conInit), c_n\colon(x_j,\Box,x_0), c_n\colon(x_0, succt,x_1) \rightarrow 
 c_n\colon(x_j, \Box, x_1) \nonumber
\end{eqnarray}
 The last BR copies the $\Box$ to every succeeding cell in the initial configuration. 
\paragraph{Transitions}
 For every left transition $\Delta (q,\sigma)= (q_j,\sigma',-1)$, the following BR:
 \begin{eqnarray}
 && c_n\colon(x_0, head,x_i), c_n\colon(x_0, \sigma, x_i), c_n\colon(x_0, state, q), c_n\colon(x_j, succ, x_i),  c_n\colon(x_0, \nonumber \\
 &&  conSucc, x_1) \rightarrow c_n\colon(x_1, head, x_j), c_n\colon(x_1,\sigma',x_i), c_n\colon(x_1, state, q_j) \nonumber 
 \end{eqnarray}
 For every right transition $\Delta (q,\sigma)=(q_j,\sigma',+1) $, the following BR:
 \begin{eqnarray}
 && c_n\colon(x_0, head,x_i), c_n\colon(x_0, \sigma, x_i), c_n\colon(x_0, state, q),  c_n\colon(x_i, succ, x_j), c_n\colon(x_0, \nonumber \\
 &&conSucc, x_1), \rightarrow  c_n\colon(x_1, head, x_j), c_n\colon(x_1,\sigma',x_i), c_n\colon(x_1, state, q_j) \nonumber 
 \end{eqnarray}

 \paragraph{Inertia}
 If in any configuration the head is at cell $i$ of the tape, then in every successor configuration, elements
 in preceding and following cells of $i$ in the tape are retained. The following two BRs ensures this:
 \begin{eqnarray}
  && c_n\colon (x_0, head, x_i), c_n\colon(x_0, conSucc, x_1), c_n\colon(x_j, succt, x_i), c_n\colon(x_0, \sigma, x_j)  \nonumber \\
  &&   \rightarrow c_n\colon(x_1, \sigma, x_j) \nonumber \\
 && c_n\colon(x_0, head, x_i), c_n\colon(x_0, conSucc, x_1), c_n\colon(x_i, succt, x_j),  c_n\colon(x_0, \sigma, x_j)  \nonumber \\
  && \rightarrow c_n\colon(x_1, \sigma, x_j) \nonumber
 \end{eqnarray}
 The rules above are instantiated for every $\sigma \in \Sigma$.
 \paragraph{Acceptance}
 An configuration whose state is $q_A$ is accepting:
 \begin{eqnarray}
  c_n\colon(x_0, state,q_A) \rightarrow c_n\colon(x_0, \texttt{rdf:type}, Accept) \nonumber
 \end{eqnarray}
If a configuration of accepting type is reached, then it can be back propagated to the initial configuration, using the following BR:
  \begin{eqnarray}
   c_n\colon(x_0, conSucc, x_1), c_n\colon(x_1, \texttt{rdf:type}, Accept) 
    \rightarrow c_n\colon(x_0, \texttt{rdf:type}, Accept) \nonumber
 \end{eqnarray}

Finally since $M$ accepts $\vec w$ iff the initial configuration is an accepting configuration, i.e.  
$QS^M_{\C} \models \exists y \ c_n\colon (y, \texttt{rdf:type}, conInit), c_n\colon(y, \texttt{rdf:type},Accept)$. 
Since there is no edge from any $c_j$ to $c_i$, for each $1 \leq i < j \leq n$, the context dependency graph for 
$QS^M_{\C}$ is acyclic, and hence $QS^M_{\C}$ is context acyclic. 
Since we reduced the word problem of 2EXPTIME DTM, which is a 2EXPTIME-hard problem, to CCQ entailment problem
over context acyclic quad-systems, it immediately follows that CCQ entailment problem over context acyclic quad-systems
is 2EXPTIME-hard. $\qed$

\end{document}